\documentclass[reqno,12pt,dvipsnames,svgnames]{amsart}

\usepackage{amsmath}
\usepackage{amssymb}
\usepackage{amsfonts}
\usepackage{amsthm}
\usepackage[foot]{amsaddr}
\usepackage{mathtools}
\usepackage{bbm}
\usepackage{tikz}
\usepackage{xcolor}
\usetikzlibrary{decorations.pathreplacing,decorations.markings}

\mathtoolsset{%
}
\usepackage{relsize}

\usepackage[utf8]{inputenc}
\usepackage[T1]{fontenc}

\usepackage{libertine}
\usepackage[libertine,libaltvw,cmintegrals
]{newtxmath}

\usepackage{dsfont}

\usepackage{mathrsfs}

\usepackage
{xcolor}
\colorlet{MyBlue}{DodgerBlue!60!Black}
\colorlet{MyGreen}{DarkGreen!85!Black}

\usepackage{fullpage}
\newcommand{\afterhead}{.}

\usepackage[font=small,labelfont=bf]{caption}
\usepackage{subfigure}

\usepackage{acronym}
\usepackage{latexsym}
\usepackage{longtable}
\usepackage{wasysym}
\usepackage{xspace}
\usepackage{enumitem}

\usepackage[authoryear,comma]{natbib}

\usepackage{hyperref}
\hypersetup{
colorlinks=true,
linktocpage=true,
pdfstartview=FitH,
breaklinks=true,
pdfpagemode=UseNone,
pageanchor=true,
pdfpagemode=UseOutlines,
plainpages=false,
bookmarksnumbered,
bookmarksopen=false,
bookmarksopenlevel=1,
hypertexnames=true,
pdfhighlight=/O,
hyperfootnotes=true,
nesting=true,
frenchlinks,
urlcolor=MyBlue!60!black,linkcolor=MyBlue!70!black,citecolor=DarkGreen!70!black, 
pagecolor=RoyalBlue,
pdftitle={},
pdfauthor={},
pdfsubject={},
pdfkeywords={},
pdfcreator={pdfLaTeX},
pdfproducer={LaTeX with hyperref}
}

\def\EMAIL#1{\email{\href{mailto:#1}{\texttt{\upshape #1}}}}

\numberwithin{equation}{section} 
\usepackage[sort&compress,capitalize,nameinlink]{cleveref}
\crefname{app}{Appendix}{Appendices}

\crefrangeformat{equation}{\upshape(#3#1#4)\textendash(#5#2#6)}

\usepackage[textwidth=30mm]{todonotes}

\usepackage{soul}
\setstcolor{red}
\sethlcolor{SkyBlue}

\newcommand{\debug}[1]{#1}

\usepackage{algorithm}
\usepackage{algpseudocode}

\theoremstyle{plain}
\newtheorem{theorem}{Theorem}

\newtheorem{corollary}[theorem]{Corollary}
\newtheorem*{corollary*}{Corollary}
\newtheorem{lemma}[theorem]{Lemma}
\newtheorem{proposition}[theorem]{Proposition}

\theoremstyle{definition}
\newtheorem{definition}[theorem]{Definition}
\newtheorem*{definition*}{Definition}

\newtheorem*{assumption*}{Assumption}
\numberwithin{theorem}{section}

\theoremstyle{remark}
\newtheorem{remark}[theorem]{Remark}
\newtheorem*{remark*}{Remark}
\newtheorem*{notation*}{Notational remark}
\newtheorem{example}[theorem]{Example}
\newtheorem*{case*}{Case}

\DeclarePairedDelimiter{\braces}{\{}{\}}
\DeclarePairedDelimiter{\bracks}{[}{]}
\DeclarePairedDelimiter{\parens}{(}{)}

\DeclarePairedDelimiter{\abs}{\lvert}{\rvert}

\DeclarePairedDelimiter{\ceil}{\lceil}{\rceil}
\DeclarePairedDelimiter{\floor}{\lfloor}{\rfloor}

\DeclarePairedDelimiterX{\braket}[2]{\langle}{\rangle}{#1,#2}

\DeclarePairedDelimiterX{\inner}[2]{\langle}{\rangle}{#1,#2}
\DeclarePairedDelimiterX{\setdef}[2]{\{}{\}}{#1:#2}

\DeclarePairedDelimiterXPP{\probof}[1]{\Prob}{(}{)}{}{%

#1}

\DeclarePairedDelimiterXPP{\exof}[1]{\Expect}{[}{]}{}{%

#1}

\DeclareMathOperator{\expo}{\debug {e}}

\newcommand{\naturals}{\mathbb{\debug N}}
\newcommand{\reals}{\mathbb{\debug R}}

\newcommand{\run}{\debug n}

\newcommand{\stime}{\debug \tau}

\DeclareMathOperator{\Expect}{\mathsf{\debug E}}

\DeclareMathOperator{\Poisson}{\mathsf{\debug{Poisson}}}
\DeclareMathOperator{\Prob}{\mathsf{\debug P}}

\DeclareMathOperator{\Var}{\mathsf{\debug{Var}}}

\newcommand{\per}{\debug t}
\newcommand{\peralt}{\debug s}
\newcommand{\Per}{\debug T}

\newcommand{\play}{\debug i}

\newcommand{\pA}{\mathrm{\debug{A}}}
\newcommand{\pB}{\mathrm{\debug{B}}}

\newcommand{\actA}{\debug a}
\newcommand{\actAalt}{\actA'}
\newcommand{\actAaltalt}{\actA''}
\newcommand{\actB}{\debug b}
\newcommand{\actBalt}{\actB'}
\newcommand{\actBaltalt}{\actB''}

\newcommand{\nactions}{\debug K}

\newcommand{\actions}{[\nactions]}

\newcommand{\eq}[1]{#1^{\ast}}

\newcommand{\potential}{\debug \Psi}
\newcommand{\potentialalt}{\potential'}
\newcommand{\potentials}{\mathcal{\debug P}}

\newcommand{\outcome}{\debug \Theta}
\newcommand{\outcomealt}{\debug \outcome'}

\newcommand{\cardNE}{\debug W}

\newcommand{\equil}{\debug{\boldsymbol{\eta}}}
\newcommand{\equilindex}{\debug \ell}
\newcommand{\rank}{\debug \Lambda}
\newcommand{\nerows}{\debug R}
\newcommand{\snerows}{\debug r}
\newcommand{\necolumns}{\debug C}
\newcommand{\snecolumns}{\debug c}
\newcommand{\tnecolumns}{\widetilde{\necolumns}}
\newcommand{\ematrix}{\debug M}
\newcommand{\greene}{\debug G}
\newcommand{\tgreene}{\widetilde{\greene}}
\newcommand{\sgreene}{\debug g}
\newcommand{\succprobrow}{\debug \rho}
\newcommand{\succrow}{\mathfrak{\debug{R}}}
\newcommand{\succprobcol}{\debug \kappa}
\newcommand{\succcol}{\mathfrak{\debug{C}}}
\newcommand{\noresample}{\mathfrak{\debug{U}}}
\newcommand{\noresampleprob}{\debug u}
\newcommand{\sampleZ}{\debug Z}
\newcommand{\sampleY}{\debug Y}
\newcommand{\sampleX}{\debug X}
\newcommand{\sigmaf}{\mathcal{\debug{F}}}
\newcommand{\nsteps}{\debug D}
\newcommand{\nstepsl}{\debug \ell}
\newcommand{\paralpha}{\debug \alpha}
\newcommand{\bern}{\debug L}
\newcommand{\constT}{\debug T}
\newcommand{\funca}{\debug \alpha}
\newcommand{\funcb}{\debug \beta}
\newcommand{\funcc}{\debug \gamma}
\newcommand{\funch}{\debug h}
\newcommand{\funcphi}{\debug \varphi}
\newcommand{\funcPhi}{\debug \Phi}
\newcommand{\profile}{\boldsymbol{\debug{x}}}
\newcommand{\zetasym}{\debug \zeta}
\newcommand{\sumprod}{\debug L}

\newcommand{\LKzeta}{\sumprod_{\nactions}(\zetasym)}
\newcommand{\harmonic}{\debug H}

\DeclareMathOperator{\BoA}{\mathsf{\debug{BoA}}}
\DeclareMathOperator{\BRD}{\mathsf{\debug{BRD}}}

\DeclareMathOperator{\NE}{\mathsf{\debug{NE}}}

\newcommand{\argdot}{\,\cdot\,}
\newcommand{\bigoh}{\mathcal{\debug{O}}}

\newcommand{\ind}{\mathds{\debug 1}}

\newcommand{\smalloh}{\debug o}

\newcommand{\ie}{i.e., }
\newcommand{\eg}{e.g., }
\newcommand{\iid}{i.i.d.\ }

\newcommand{\wlg}{w.l.o.g, }

\DeclareMathOperator*{\argmin}{arg\,min} 

\newcommand{\coupont}{\debug T}
\newcommand{\genact}{\debug k}
\newcommand{\evented}{\mathcal{\debug A}_{\varepsilon,\delta}}

\newacro{BRD}{best response dynamic}
\newacro{bRD}{better response dynamic}
\newacro{NE}{Nash equilibrium}
\newacroplural{NE}[NE]{Nash equilibria}
\newacro{PNE}{pure Nash equilibrium}
\newacroplural{PNE}[PNE]{pure Nash equilibria}
\newacro{MNE}{mixed Nash equilibrium}
\newacroplural{MNE}[MNE]{mixed Nash equilibria}
\newacro{PFNE}{prior-free Nash equilibrium}
\newacroplural{PFNE}[PFNE]{prior-free Nash equilibria}

\newacro{KKT}{Karush\textendash Kuhn\textendash Tucker}

\newacro{FIP}{finite improvement property}
\newacro{CLT}{central limit theorem}

\newacro{OPG}{ordinal potential game}
\newacro{SOP}{strictly ordinal potential}
\newacro{BoA}{basin of attraction}

\newacro{ODE}{ordinary differential equation}

\begin{document}

\title[Random Ordinal Potential Games]{Basins of Attraction in Two-Player Random Ordinal Potential Games}

\author{Andrea Collevecchio$^1$}
\address{$^1$ School of Mathematics, and Monash Data Futures Institute, Monash University, Melbourne, Victoria 3800, Australia.}
\EMAIL{andrea.collevecchio@monash.edu}

\author{Hlafo Alfie Mimun$^2$}
\address{$^2$ Dipartimento di Economia e Finanza, Luiss University, Viale Romania 32, 00197 Roma, Italy.}
\EMAIL{hmimun@luiss.it}

\author{Matteo Quattropani$^3$}
\address{$^3$ Dipartimento di Matematica ``Guido Castelnuovo'', Sapienza Universit\`a di Roma, Piazzale Aldo Moro 5, 00185  Roma, Italy.}
\EMAIL{matteo.quattropani@uniroma1.it}

\author{Marco Scarsini$^2$}
\EMAIL{marco.scarsini@luiss.it}

\begin{abstract}
We consider the class of two-person ordinal potential games where each player has the same number of actions $\nactions$.
It is well-known that each game in this class admits at least one \acl{PNE} and the \acl{BRD} converges to one of these \aclp{PNE};
which one depends on the starting point.
So, each \acl{PNE} has a \acl{BoA}.

We pick uniformly at random one game from this class and we study the joint distribution of the sizes of the basins of attraction.
We provide an asymptotic exact value for the expected \acl{BoA} of each \acl{PNE}, when the number of actions $\nactions$ goes to infinity.
  


\end{abstract}

\maketitle


\section{Introduction}
\label{se:introduction}

There exists a huge literature dealing  with dynamics in games and their ability to reach a \ac{PNE}. 
Various impossibility results have been proved, showing that no simple adaptive dynamics can reach a \ac{PNE} in every game that has one.
See, for instance, the papers in the collection by \citet{HarMas:WS2013} and the more recent articles by \citet{MilPapPilSpe:PNAS2023,HakMilPapPil:SAGT2024}, among many others. 
To determine the actual relevance of these impossibility results, \citet{JohSavScoTar:arXiv2023}  considered the class of ordinal games with fixed number of players and fixed action sets for each player and looked at the frequency of games for which adaptive dynamics fail to converge to a Nash equilibrium even if it exists. 
They proved that this frequency vanishes as the number of players diverges, provided the number of actions for each player do not grow too fast. 
Although their results are fully combinatorial, they can be framed within the literature of games with random payoffs.  
Most papers about games with random payoffs focus on the number of \acp{PNE} and its distribution. 
More recently, the attention has shifted to the behavior of adaptive dynamics in games with random payoffs. 
Our paper can be seen as part of this stream of literature, in the sense that we consider games with random payoffs and we study \ac{BRD} for a class of such games. 
In particular, we focus on ordinal potential games. 
It is well known that all games in this class admit \acp{PNE} and \ac{BRD} converges to one of these equilibria.
The issue we want to study in this paper is to which equilibria is a \ac{BRD} more likely to converge.

\subsection{Our contribution}
\label{suse:contribution}
In this paper we consider a particular class of two-player normal-form games, called \aclp{OPG}. 
These games always admit \aclp{PNE}, and the common learning dynamics converge to a \acl{PNE}.
In this class of games, all the strategic features of the game are embedded in the potential function, \ie it is sufficient to know the potential function to pinpoint  the \aclp{PNE} of an \acl{OPG}.
Therefore, if we want to pick an \acl{OPG} at random, we can put a suitable probability measure on the space of potentials. 
For the sake of simplicity, we will assume that the two players have the same number of actions $\nactions$.
We will make the strong assumption that all outcomes in the potential  are strictly comparable, \ie no equivalences are allowed. 
With this assumption, a random \acl{OPG} is equivalent to a random permutation on of the first $\nactions^{2}$ integers.

The structure of the game implies that, for a given realization of the potential, if we start at a fixed action profile and we perform \acl{BRD} starting with the first player (i.e., the player choosing the row), then we reach a \acl{PNE}.

The goal of this paper is to determine the probability of reaching the  \acl{PNE} with the smallest or second smallest or  third smallest potential, etc., when also the starting action profile is chosen uniformly at random.
This translates into finding the expected size of the basin of attraction of each \acl{PNE} of the game. 
Solving this problem exactly for a specific choice of $\nactions$ is a quite complicate endeavor. 
We adopt an asymptotic approach and study the problem as $\nactions\to\infty$.
We use a result of \citet{MimQuaSca:GEB2024} according to which the fraction of the number of \aclp{PNE} over $\nactions$ converges to $1/2$.
Based on this, we see that the distribution of the rescaled ranking of \acl{PNE} found by the \acl{BRD} converges to an absolutely continuous distribution supported on the interval $[0,1/2]$, which we characterize explicitly.


\subsection{Related literature}

As mentioned before, most of the initial contributions on random games, starting with \citet{Gol:AMM1957}, assumed \iid payoffs with a continuous distribution, and studied the existence of \aclp{PNE} and the distribution of their random number.
Two papers proved that, under the above hypotheses, the number of \aclp{PNE} converges in distribution to a Poisson distribution with parameter $1$ as either the number of players diverges \citep{ArrGolGor:AP1989} or the number of actions for each player diverges \citep{Pow:IJGT1990}.

\citet{RinSca:GEB2000} retained the assumption of independence for the random payoff vectors corresponding to different action profiles, but allowed for dependence of payoffs within the same action profile (modeled by an exchangeable multinormal distribution).
They proved an interesting phase transition in the correlation coefficient as either the number of players or the number of actions diverges.
When correlation is negative, the number of \aclp{PNE} converges to $0$, when the correlation is null, it converges to $\Poisson(1)$, when the correlation is positive, it diverges and a central limit theorem holds. 

More recently, several papers started studying the behavior of \acl{BRD} in games with random payoffs.
\citet{AmiColScaZho:MOR2021} 
studied games with many players and two actions for each player.
They assumed \iid payoffs, but allowed distributions to have atoms.
They proved that the number of \aclp{PNE} diverges exponentially fast in a way that depends on the parameter that corresponds to the size of the atoms.
Moreover, they studied the behavior of \acl{BRD} and proved some phase transitions in the above parameter.
To achieve their results, they used some percolation techniques.
Their results were recently generalized by \citet{ColNguZho:arXiv2024}.
\citet{AmiColHam:ORL2021} compared the asymptotic behavior of \acl{BRD} and \acl{bRD} for two-player random games with \iid payoffs with a continuous distribution.
\citet{HeiJanMunPanScoTarWie:IJGT2023} studied the role of the playing sequence in the convergence to equilibrium of \acl{BRD}.
\citet{JohSavScoTar:arXiv2023} studied the model of \citet{AmiColScaZho:MOR2021}  with no atoms and more than two actions per player, and allow this number to grow, although at a lower rate than the number of players.
For an extensive list of  contributions on random games, we refer to the papers by \citet{AmiColScaZho:MOR2021} and \citet{JohSavScoTar:arXiv2023}.

\citet{PraTar:arXiv2024} introduced the concept of satisficing equilibria, which corresponds to the idea that players play one of their $k$ best pure actions, not necessarily the best.
The proved that  in almost all games, there exist satisficing equilibria in which all but one player use best response and the remaining player plays at least a second-best action.
Various aspects of games with random payoffs were recently studied in \citet{HarRinWei:AAP2008,HolMarMar:PTRF2019,AloRudYar:arXiv2021,WieHei:DGA2022,PeiTak:IJGT2023,FlePreSuo:MOR2023,GarZil:MOR2023,BhaKarPodRoy:EJP2023,NewSaw:2024,SepZil:arXiv2024}.

In a paper that, similarly to ours, focuses on the quality of different equilibria, \citet{CanSakLinVarPil:arXiv2024} studied techniques that can be used to steer players towards one of several equilibria in a game in some specific class, \eg coordination games. 
Their approach does not consider games with random payoffs.
\citet{HakMilPapPil:SAGT2024} studied 
the problem of computing the asymptotic behavior of noisy replicator dynamics as a limit distribution
over the sink equilibria of a game.
The authors generated games with random utilities to carry out randomized experiments of running noisy replicator dynamics.
\citet{KiaPedRyaSmiZab:INFORMSJOfrth}  studied a variation of fictitious play called Monte Carlo fictitious play (MCFP) and proved that, with the aid of an auxiliary game, it converges to a \acl{PNE} in games of common interest. 

The problem of learning equilibria in repeated games with a large number of actions was studied by \citet{CheCheFosHaz:arXiv2024} in the setting of AI Safety via Debate.

The idea of ordinal games goes back at least to \citet{Fis:IJGT1978}, who used the term ``stochastic dominance games'' to identify them. 
They were then studied by several authors \citep[see, \eg][]{CruSim:JOTA2000,Bon:LOFT2008,GafSal:2015}.
The computational complexity of their equilibria was recently studied by \citet{Con:AAAAICAI2024}.
\citet{PieBic:arXiv2024} proposed a new solution concept, called $\alpha$-Rank-collections, for games where  ordinal preferences are known, but there is uncertainty concerning players' cardinal utilities.

Ordinal potential games were introduced and studied by \citet{MonSha:GEB1996}, and then characterized by \citet{VooNor:GEB1997} and \citet{NorPat:2001}.


\section{The game}
\label{se:game}

\subsection{Ordinal games}
\label{suse:ordinal-games}

Throughout the paper, for every $\run\in\naturals$, we let $\bracks{\run}\coloneqq\braces{1,\dots,\run}$.
We study two-person normal form finite ordinal games where each player $\play\in\braces{\pA,\pB}$ has the same action set $\actions$ and a preference relation $\prec_{\play}$ over the outcomes of the game. 
We assume these preferences to be strict, \ie for all pairs of outcomes $\outcome,\outcomealt$, either $\outcome\prec_{\play}\outcomealt$ or $\outcomealt\prec_{\play}\outcome$.

A strategy profile $(\eq\actA,\eq\actB)$ is a \acfi{NE}\acused{NE} of the game if, for all $\actA,\actB\in\actions$, we have 
\begin{equation*}
\label{eq:Nash}
(\actA,\eq\actB)\prec_{\pA}(\eq\actA,\eq\actB) 
\quad\text{and}\quad
(\eq\actA,\actB)\prec_{\pB}(\eq\actA,\eq\actB).
\end{equation*}

A two-person normal form game is called \acfi{SOP}\acused{SOP} if there exists a \emph{potential function} $\potential \colon \actions\times\actions\to\reals$ such that, for each player $\play\in\braces{\pA,\pB}$, we have
\begin{equation}
\label{eq:ordinal-potential} 
\begin{split}
\outcome(\actA,\actB) &\prec_{\pA} \outcome(\actAalt,\actB)  
\iff
\potential(\actA,\actB) > \potential(\actAalt,\actB),\\
\outcome(\actA,\actB) &\prec_{\pB} \outcome(\actA,\actBalt)  
\iff
\potential(\actA,\actB) > \potential(\actA,\actBalt).
\end{split}
\end{equation}
To allow us to express asymptotic results in a simple way, the inequality in the potential is reversed with respect to the preference orders:
given the action of the other player, an outcome with a lower potential is better than an outcome  with a higher potential. 
For our class of \ac{SOP} games, \wlg the potential function can be chosen to be $1$-$1$ with values in $\bracks{\nactions^{2}}$.
In other words, the outcome whose potential is $1$ is the best, the one whose potential is $2$ is the second best, etc, all the way to the worst, whose potential is $\nactions^{2}$.
The class of  $1$-$1$ potentials with values in $\bracks{\nactions^{2}}$ will be denoted by $\potentials_{\nactions}$.
In the sequel we identify the function $\potential\in\potentials_{\nactions}$ with the $\actions\times\actions$ matrix of its values.

Two \ac{SOP} games having the same potential are strategically equivalent, \ie they have the same set of \aclp{NE}.
Any \ac{SOP} game is strategically equivalent to a game where $\prec_{\pA} \equiv \prec_{\pB}$.
The potential identifies the set of \acl{NE} of any \ac{SOP} game.
Therefore, we identify the equivalence class of strategically equivalent \acp{SOP} with their  potential $\potential\in\potentials_{\nactions}$.
Each \ac{NE} of a \ac{SOP} game is a local minimum of its potential.
The set of \acp{NE} of $\potential\in\potentials_{\nactions}$ will be denoted by $\NE_{\nactions}$.
Its cardinality $\abs*{\NE_{\nactions}}$ will be denoted by $\cardNE_{\nactions}$.
The set
\begin{equation}
\label{eq:set-equilibrium-potential} 
\eq\potential \coloneqq \braces*{\potential(\equil) \colon \equil \in \NE_{\nactions}} \subset \bracks*{\nactions^{2}}
\end{equation}
is the set of potential values corresponding to equilibrium profiles (in the sequel, they will be called equilibrium potentials).
The \acp{NE} in a \ac{SOP} can be ordered according to their potential, so that the equilibrium $\equil_{1}$ is the one with the smallest potential, $\equil_{2}$ is the one with the second smallest potential, etc.
The ranking of equilibrium $\equil$ is denoted by $\rank(\equil)$.


\section{Random \acl{SOP} games}
\label{se:random-SOP}

We now study random \ac{SOP} games. 
To achieve this, for every fixed $\nactions$, we consider a uniform distribution over the set $\potentials_{\nactions}$. 
A simple way to generate a uniform distribution on $\potentials_{\nactions}$ is to draw $\nactions^{2}$ \iid random variables with a continuous distribution function (\wlg a uniform distribution on $[0,1]$) and to consider their rankings, which, with probability $1$, are all different.

A random \ac{SOP} game has a random number of \acp{NE}, which is a random variable with values in $\actions$, as each column and raw can contain at most one \acp{NE}.
This random number is positive, since the potential always has a global minimum and cannot be larger than $\nactions$ because the preferences are strict, \ie the values that the potential can take are all distinct. 

It is easy to compute the average number of \acp{NE} in a random \ac{SOP}.
Each action profile is a \ac{NE} with probability $1/(2\nactions-1)$ because it is a \ac{NE} if the value of its potential is larger than all the values on the same row and column.
Therefore, the expected number of \acp{NE} is 
\begin{equation}
\label{eq:Exp-card-NE} 
\Expect\bracks*{\cardNE_{\nactions}} = \frac{\nactions^{2}}{2\nactions-1}.
\end{equation}
The following concentration result was proved by  {\citet[corollary~3.3]{MimQuaSca:GEB2024}}.

\begin{theorem}
[{\citet[corollary~3.3]{MimQuaSca:GEB2024}}]
\label{th:Expect-card-NE-conc}
For all $\delta>0$,
\begin{equation}  
\lim_{\nactions\to\infty} \Prob\parens*{\abs*{\frac{\cardNE_{\nactions}}{\nactions}-\frac{1}{2}}<\delta} = 1.
\end{equation}
\end{theorem}

In what follows, we will often use the classical asymptotic notation: for positive sequences $f_n$ and $g_n$ we write $f_n=\bigoh(g_n)$ to mean that the ratio $f_n/g_n$ is asymptotically bounded; similarly, we write $f_n=\smalloh(g_n)$ to mean that the limiting ratio $f_n/g_n$ is zero. Finally, $f_n\sim g_n$ denotes the fact that the limiting ratio is one.



\section{\acl{BRD}}
\label{se:BRD}
 
The \acfi{BRD}\acused{BRD} is a learning algorithm taking as input a potential $\potential$ and a starting action profile $\parens{\actA_{0},\actB_{0}}$. 
For each $\per\ge0 $ we consider the process $\BRD(\per)$ on $\actions \times \actions$ such that 
\begin{align}
\label{eq:BRD}
\BRD(0) &= \parens{\actA_{0},\actB_{0}} 
\intertext{and, if $\BRD(\per)=\parens{\actAalt,\actBalt}$, then, for $\per$ even,}
\BRD(\per+1) &= \parens{\actAaltalt,\actBalt},
\intertext{where 
$\actAaltalt\in\argmin_{\actA\in\actions}\potential(\actA,\actBalt) \setminus \braces*{\actAalt}$, if the latter set is not empty, otherwise}
\BRD(\per+1) &= \BRD(\per); 
\intertext{for $\per$ odd,}
\BRD(\per+1) &= \parens{\actAalt,\actBaltalt},
\intertext{where 
$\actBaltalt\in\argmin_{\actB\in\actions}\potential(\actAalt,\actB)  \setminus \braces*{\actBalt}$, if the latter set is not empty, otherwise}
\BRD(\per+1) &= \BRD(\per).
\end{align}
It is easy to see that, when, for some 
positive  $\hat\per$, we have
\begin{equation}
\label{eq:BRD-PNE}
\BRD(\hat\per) = \BRD(\hat\per+1) = 
\parens{\eq\actA,\eq\actB},
\end{equation}
then $\BRD(\per)=\parens{\eq\actA,\eq\actB}$ for all $\per\ge\hat\per$ and  $\parens{\eq\actA,\eq\actB}$ is a \ac{NE} of the game.  

The fact that the game is \ac{SOP} implies that a \ac{NE} is always reached by the \ac{BRD}.
When the algorithm visits an action profile for the second time, that profile is a \ac{NE}.
Notice that a \ac{BRD} never visits a row or column more than twice (once by the row player and once by the column player), so it reaches a \ac{PNE} in at most $2\nactions$ steps. For this reason, we will often employ the notation $\BRD(2\nactions)$ to denote the action profile in which the \ac{BRD} will eventually get absorbed in.

Once a starting point $\BRD(0) = \parens{\actA_{0},\actB_{0}}$ is chosen, the \ac{BRD} will reach (deterministically) one \ac{NE}.
For each \ac{NE} $(\eq\actA,\eq\actB)$, we define its \acfi{BoA}\acused{BoA} as follows:
\begin{equation}
\label{eq:BoA} 
\BoA(\eq\actA,\eq\actB)
\coloneqq
\braces*{(\actA,\actB) \colon \text{if }\BRD(0)=(\actA,\actB), \text{ then }
 \BRD(2\nactions)=(\eq\actA,\eq\actB)}
\end{equation}

\begin{remark}
\label{re:BoA(a,b)-a}
Given the way the process $\BRD(\argdot)$ is defined,  we have that $(\actA,\actB)\in\BoA(\eq\actA,\eq\actB)$ implies $(\actAalt,\actB)\in\BoA(\eq\actA,\eq\actB)$ for all $\actAalt\in\actions$.  
This is because the very first step of the \ac{BRD} is along the row.

\end{remark}

\begin{remark}
\label{le:permutation-potential}
Notice that if the potential $\potentialalt$ is obtained from $\potential$ by permuting rows and columns, then the \acp{PNE} of\, $\potentialalt$ are just the corresponding permutations of the \acp{PNE} of\, $\potential$.
Moreover, the \aclp{BoA} of the \acp{PNE} in $\potentialalt$ are obtained by permuting the columns of the corresponding \aclp{BoA} in $\potential$.
\end{remark}

\begin{example}
\label{ex:permutation-potential}  
We now consider the realization of a \ac{SOP} and its rearrangement, obtained by permuting rows and columns.  

\begin{align*}
\begin{array}{c}
\textcolor{green}{ \begin{array} {c@{\hspace{16pt}}c@{\hspace{16pt}}c@{\hspace{16pt}}c@{\hspace{16pt}}c@{\hspace{16pt}}c@{\hspace{16pt}}c@{\hspace{16pt}}c@{\hspace{16pt}}c@{\hspace{16pt}}c@{\hspace{3pt}}}
4 & 1 & 5 & 2 & 2 & 4 & 1 & 1 & 5 & 4   
\end{array}}\\[5pt]
\begin{pmatrix}
100 & 56 & 43 & 32 & 26 & 24 & 12 & 55 & 39 & 40 \\
77 & 83 & 82 & 48 & 29 & 79 & 44 & 92 & 53 & 95 \\
97 & 3 & 28 & 23 & 57 & 30 & 91 & 17 & 41 & 89 \\
21 & 63 & 99 & 73 & 59 & \textcolor{green}{4} & 25 & 49 & 85 & 9 \\
42 & 66 & 20 & 72 & 27 & 54 & 68 & 98 & 71 & 67 \\
31 & 15 & 6 & 50 & 90 & 18 & 70 & 81 & 84 & 34 \\
96 & 16 & \textcolor{green}{5} & 38 & 78 & 65 & 47 & 36 & 8 & 60 \\
69 & 64 & 86 & 10 & \textcolor{green}{2} & 46 & 61 & 35 & 13 & 14 \\
45 & \textcolor{green}{1} & 62 & 74 & 19 & 52 & 7 & 11 & 51 & 94 \\
37 & 75 & 88 & 80 & 33 & 76 & 22 & 87 & 58 & 93 
\end{pmatrix} 
\end{array}
\quad
\begin{array}{c}
\textcolor{green}{ \begin{array} {c@{\hspace{16pt}}c@{\hspace{16pt}}c@{\hspace{16pt}}c@{\hspace{16pt}}c@{\hspace{16pt}}c@{\hspace{16pt}}c@{\hspace{16pt}}c@{\hspace{16pt}}c@{\hspace{16pt}}c@{\hspace{3pt}}}
1 & 2 & 4 & 5 & 1 & 5 & 4 & 2 & 1 & 4  
\end{array}}\\[5pt]
\begin{pmatrix}
\textcolor{green}{1} & \textcolor{brown}{19} & 52 & 62 & \textcolor{red}{7} & 51 & 94 & 74  & \textcolor{red}{11} & 45   \\
64 & \textcolor{green}{2} & 46 & 86 & 61 & \textcolor{brown}{13} & \textcolor{brown}{14} & \textcolor{red}{10} & 35 & 69  \\
\textcolor{red}{3} & 57 & 30 & 28 & 91 & 41 & 89  & 23 & \textcolor{brown}{17} & 97\\
63 & 59 & \textcolor{green}{4} & 99 & 25 & 85 & \textcolor{red}{9} & 73 & 49 & \textcolor{red}{21}  \\
\textcolor{brown}{16} & 78 & 65 & \textcolor{green}{5} & 47 & \textcolor{red}{8} & 60 & 38 & 36 & 96  \\
\textcolor{brown}{15} & 90 & \textcolor{brown}{18} & \textcolor{red}{6} & 70 & 84 & 34 & 50 & 81 & 31 \\
56 & 26 & 24 & 43 & \textcolor{red}{12} & 39 & 40 & 32 & 55 & 100 \\
66 & 27 & 54 & \textcolor{red}{20} & 68 & 71 & 67 & 72 & 98 & 42 \\
75 & 33 & 76 & 88 & 22 & 58 & 93 & 80 & 87 & 37 \\
83 & 29 & 79 & 82 & 44 & 53 & 95 & 48 & 92 & 77 
\end{pmatrix} 
\end{array}
\end{align*}
The potential on the right is obtained by permuting some rows and columns of the potential on the left.
The green numbers in the above matrices indicate the potential equilibria.
The numbers above the matrices indicate the potential of the equilibrium to which the column is attracted.
The meaning of the other colors will be explained in     \cref{re:incremental-construction-permutation} below.
\end{example}

Our goal is to study the  \acp{BoA} of the different \acp{NE}.
In particular, we will focus on their size. 
An exact analysis for fixed $\nactions$ is quite cumbersome, but we have some interesting asymptotic results. 
We remind the reader that $\equil_{i}$ is the equilibrium with the $i$-th highest potential.

\begin{theorem}
\label{th-BoA-espilon}
For all $\varepsilon\in[0,1/2)$, we have
\begin{equation}
\label{eq:BoA-epsilon}
\lim_{\nactions\to\infty}
\Expect\bracks*{\frac{1}{\nactions}\abs*{\BoA_{\nactions}(\equil_{\floor{\varepsilon\nactions}})}}
=
\funcphi(\varepsilon)\coloneqq\exp\braces*{\sqrt{1-2\varepsilon}}.
\end{equation}
\end{theorem}

The following corollary shows how the ranking $\rank$ of the equilibria reached by the \ac{BRD} behaves asymptotically. 

\begin{corollary}
\label{co:BoA-epsilon}  
Let $\BRD_{\nactions}(0)$ be chosen uniformly at random on $\nactions\times\nactions$.
For all $\varepsilon\in[0,1/2)$, we have
\begin{equation}
\label{eq:BoA-epsilon-DF}  
\lim_{\nactions\to\infty}
\Prob\parens*{\rank(\BRD(2\nactions)) \le \varepsilon\nactions} 
= \funcPhi(\varepsilon) \coloneqq
(1-\sqrt{1-2\varepsilon})\exp\braces*{\sqrt{1-2\varepsilon}}.
\end{equation}
Moreover,
\begin{equation}
\label{eq:BoA-epsilon-Expect}  
\lim_{\nactions\to\infty}
\Expect\bracks*{\frac{\rank(\BRD(2\nactions))}{\nactions}} 
=\expo-\frac{5}{2}
\approx 0.21.
\end{equation}
\end{corollary}

%

The value $\rank(\BRD(2\nactions))$ represents the relative (with respect to $\nactions$) ranking of the equilibrium reached by the \ac{BRD}.
This random variable is supported on the interval $[0,1/2]$.
\cref{fi:cdf-pdf-BRD} shows the plot of its density function $\funcphi$ and its distribution  function $\funcPhi$.

\begin{figure}[H]
\centering
\includegraphics[width=7cm]{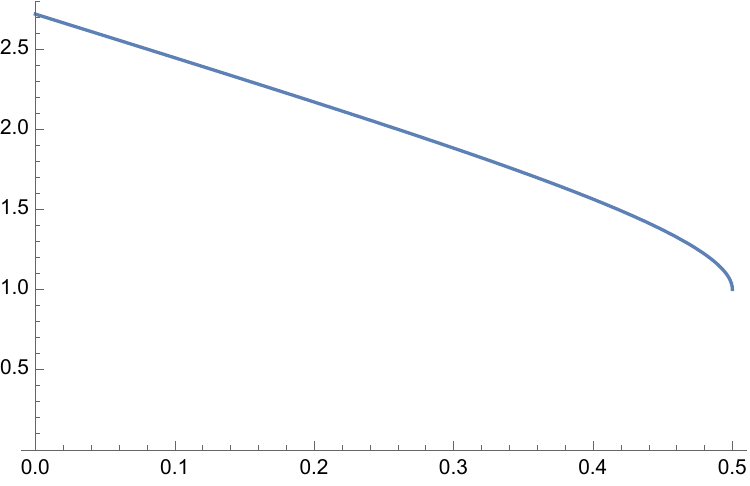}
\qquad
\includegraphics[width=7cm]{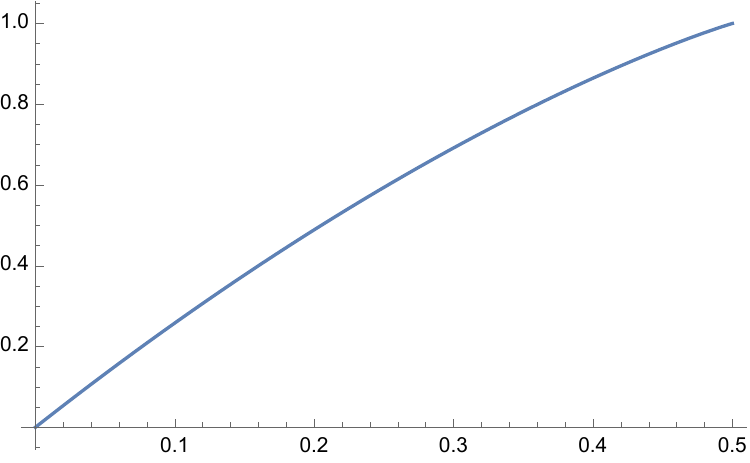}    
\caption{Plot of the functions 
$\funcphi(\argdot)$ (left) and 
$\funcPhi(\argdot)$ (right).}
    \label{fi:cdf-pdf-BRD}
\end{figure}


The next theorem shows that the potential of the worst \ac{PNE}  cannot be much larger than $\nactions\log\nactions$. 

\begin{theorem}
\label{th:BoA-rank-potential} 
For all $\delta>0$, we have
\begin{equation}
\label{eq:BoA-rank-potential}  
\lim_{\nactions\to\infty}
\Prob\parens*{\frac{\potential(\equil_{\cardNE_{\nactions}})}{\nactions \log \nactions} <1+ \delta} 
= 1.
\end{equation}
\end{theorem}


\section{The incremental construction}
\label{se-incremental-construction}

\cref{th-BoA-espilon} will be proved using what we call the \emph{incremental construction} of the game.
This construction provides the potential of a random \ac{SOP} game that does not have a uniform distribution, but has the same set of equilibrium potentials that a uniformly distributed potential has. 

For a fixed integer $\nactions$, we will construct a random potential function  $\potential\in\potentials_{\nactions}$ by adding entries sequentially according to the  algorithm described below.

Set $\necolumns_0=\nerows_0=0$ and, for $\per\in\bracks*{\nactions^{2}}$, 

\begin{enumerate}[label=(\alph*), ref=(\alph*)]
\item 
we call $\nerows_{\per}$ the number of non-empty rows, $\necolumns_{\per}$ the number of non-empty columns, and $\greene_{\per}$ the number of green entries after adding the first $\per$ entries of $\potential$;

\item 
we call $\ematrix_{\per}$ the sub-matrix of $\potential$ composed of rows $\bracks*{\nerows_{\per}}$ and columns $\bracks*{\necolumns_{\per}}$;

\item 
we call $\succrow_{\per}$ a Bernoulli random variable such that 
\begin{equation}
\label{eq:success-row-prob} 
\Prob(\succrow_{\per}=1)=
\succprobrow_{\per} 
\coloneqq
\frac{(\nactions-\nerows_{\per-1})\nactions}{\nactions^{2}-\per-1};
\end{equation}

\item 
we call $\succcol_{\per}$ a Bernoulli random variable such that 
\begin{equation}
\label{eq:success-col-prob} 
\Prob(\succcol_{\per}=1)=
\succprobcol_{\per} 
\coloneqq
\frac{(\nactions-\necolumns_{\per-1})}{\nactions}.
\end{equation}
\end{enumerate}

The incremental construction will be described by \cref{al:incremental-construction}.

\begin{algorithm} [H]
\caption{Incremental construction}
\label{al:incremental-construction}
\begin{enumerate}
\item 
Set $\potential(1,1)=1$.

\item 
Color the entry $(1,1)$ green.

\item 
For $\per\in\braces*{1,\dots,\nactions^{2}}$, the $\per+1$-th entry is added as follows:

\begin{enumerate}[label=(\alph*), ref=(\alph*)]
\item 
If $\succrow_{\per+1}=1$, then set $\nerows_{\per+1}=\nerows_{\per}+1$.

\begin{enumerate}[label=(\alph{enumii}.\roman*), ref=(\alph{enumii}.\roman*)]
\item 
If $\succcol_{\per+1}=1$, then set $\necolumns_{\per+1}=\necolumns_{\per}+1$, 
$\potential(\nerows_{\per+1},\necolumns_{\per+1})=\per+1$,
and color $(\nerows_{\per+1},\necolumns_{\per+1})$ green.

\item 
If $\succcol_{\per+1}=0$, then set $\necolumns_{\per+1}=\necolumns_{\per}$, sample $\sampleZ_{\per+1}$ uniformly at random in $\bracks*{\necolumns_{\per}}$, and set $\potential(\nerows_{\per+1},\sampleZ_{\per+1})=\per+1$.
\end{enumerate}

\item 
\label{it:R-t+1=0}
If $\succrow_{\per+1}=0$, then set $\nerows_{\per+1}=\nerows_{\per}$
and draw one entry uniformly at random among the empty entries in the rows $\braces*{1,\dots,\nerows_{\per}}$.
Call this entry $(\sampleX_{\per+1},\sampleY_{\per+1})\in[\nerows_{\per}]\times[\nactions]$.
\begin{enumerate}[label=(\alph{enumii}.\roman*), ref=(\alph{enumii}.\roman*)]
\item 
If $(\sampleX_{\per+1},\sampleY_{\per+1})\not\in\ematrix_{\per}$, then set $\necolumns_{\per+1}=\necolumns_{\per}+1$ and  $\potential(\sampleX_{\per+1},\necolumns_{\per+1})=\per+1$.

\item 
If $(\sampleX_{\per+1},\sampleY_{\per+1})\in\ematrix_{\per}$, then  set $\necolumns_{\per+1}=\necolumns_{\per}$ and $\potential(\sampleX_{\per+1},\sampleY_{\per+1})=\per+1$.

\end{enumerate}

\end{enumerate}

\item 
The output of the algorithm will be called $\potential$.

\end{enumerate}  

\end{algorithm}

In the incremental construction  call $\potential^{\per}$ the set of all entries added at times $\peralt \le \per$ and 
$\sigmaf_{\per}\coloneqq \sigma(\potential^{\per})$.
So $\braces*{\sigmaf_{\per}}_{\per\in\braces{0,\dots,\nactions^{2}}}$ is the natural filtration.

\begin{remark}
\label{re:incremental-construction-permutation}  
The potential on the right in \cref{ex:permutation-potential} is a possible outcome of an incremental construction. In that example, the incremental construction ends right after having added the entry $21$. The numbers in green represent Nash equilibria; those in red represent steps of the algorithm in which a new row (or a new column, but not both) is created; the numbers in brown represent the steps on the algorithm in which the number of row and columns of $\ematrix$ do not change.
\end{remark}


\section{Asymptotic analysis of the incremental construction}
\label{se:asymptotics-incremental-construction}

\subsection{$\varepsilon$-stopped incremental construction}
\label{suse:epsilon-stopped-IC}

\begin{definition}
\label{de:epsilon-IC}    
Fix $\varepsilon\in(0,1/2)$. An \emph{$\varepsilon$-stopped incremental construction} of a potential $\potential\in\potentials_{\nactions}$ is obtained as follows: perform the incremental construction of \cref{al:incremental-construction} up to the stopping time $\stime_{\varepsilon\nactions}$ defined as
\begin{equation}
\label{eq:epsilon-time}
\stime_{\varepsilon\nactions} \coloneqq \inf \braces*{\per \ge 0 \colon \greene_{\per} = \floor{\varepsilon\nactions}}.
\end{equation}
After $\stime_{\varepsilon\nactions}$ continue the construction by placing the remaining integers $\stime_{\varepsilon\nactions}+1, \dots, \nactions^{2}$ uniformly at random on the remaining empty entries.
\end{definition}

Notice that the set of best $\floor{\varepsilon\nactions}$ equilibria and the value of their potential are measurable with respect to $\sigmaf_{\stime_{\varepsilon\nactions}}$.

\begin{proposition}
\label{pr:asymptotic-epsilo-IC} 
Fix $\varepsilon\in(0,1/2)$. 
\begin{enumerate}[label={\rm(\alph*)}, ref=(\alph*)]
\item 
\label{it:pr:asymptotic-epsilo-IC-a}
The stopping time $\stime_{\varepsilon\nactions}$ satisfy
\begin{equation}
\label{eq:stopping-time-concentrates}
\frac{\stime_{\varepsilon\nactions}}{\nactions} 
\xrightarrow[\nactions\to\infty]{\Prob}
\frac{1}{2}\log\parens*{\frac{1}{1-2\varepsilon}}.
\end{equation}

\item 
\label{it:pr:asymptotic-epsilo-IC-b}
The quantities $\nerows_{\stime_{\varepsilon\nactions}}$ and $\necolumns_{\stime_{\varepsilon\nactions}}$  satisfy
\begin{equation}
\label{eq:random-matrices-concentrates}
\parens*{\frac{\nerows_{\stime_{\varepsilon\nactions}}}{\nactions}, \frac{\necolumns_{\stime_{\varepsilon\nactions}}}{\nactions}
}
\xrightarrow[\nactions\to\infty]{\Prob}
\parens*{1 - \sqrt{1-2\varepsilon},1 - \sqrt{1-2\varepsilon}
}.
\end{equation}
\end{enumerate}
\end{proposition}

The proof of \cref{pr:asymptotic-epsilo-IC} requires the introduction of a new process, $(\tnecolumns_{\per})_{\per\ge 0}$, defined below.

In \cref{al:incremental-construction}\ref{it:R-t+1=0} let $\sampleY_{\per+1}$ be sampled as follows: draw a cell at random in the set of rows $\braces*{1,\dots,\nerows_{\per}}$. 
If the cell is not empty, draw another cell, and keep drawing until an empty cell is selected.
Call $\noresample_{\per+1}$ a Bernoulli random variable that takes value $1$ if no resamples are needed. 
We have
\begin{equation}
\label{eq:no-resamples} 
\Prob(\noresample_{\per+1}=1)=
\noresampleprob_{\per+1} = 
1- \frac{\per}{\nerows_{\per}\nactions}.
\end{equation}

The process $\tnecolumns_{\per}$ is constructed as follows:
\begin{enumerate}[label=(\alph*), ref=(\alph*)]
\item 
\label{it:init}
Set $\tnecolumns_{0}=0$.

\item 
\label{it:srow1}
If $\succrow_{\per+1}=1$, then set $\tnecolumns_{\per+1}-\tnecolumns_{\per}=\necolumns_{\per+1}-\necolumns_{\per}$.

\item 
\label{it:srow0-noresample1}
If $\succrow_{\per+1}=0$ and $\noresample_{\per+1}=1$, then 
set $\tnecolumns_{\per+1}-\tnecolumns_{\per}=\necolumns_{\per+1}-\necolumns_{\per}$.

\item 
\label{it:srow0-noresample0}
If $\succrow_{\per+1}=0$ and $\noresample_{\per+1}=0$, then 
set $\tnecolumns_{\per+1}=\necolumns_{\per}$.
\end{enumerate}

Call $\nsteps_{\per}$ the number of steps where \ref{it:srow0-noresample0} occurs up to time $\per$.
It follows that, a.s.,
\begin{equation}
\label{eq:bound-steps} 
0 \le \necolumns_{\per} - \tnecolumns_{\per} \le \nsteps_{\per}.
\end{equation}
Moreover,
\begin{equation}
\label{eq:P-nsteps}  
\Prob\parens*{\nsteps_{\per+1} = \nsteps_{\per}+1} = 
\Prob\parens*{\succrow_{\per+1}=0,  \noresample_{\per+1}=0} =
\frac{\nerows_{\per}\nactions-\per}{\nactions^{2}-\per}
\frac{\per}{\nerows_{\per}\nactions}
\le \frac{\per}{\nactions^{2}},
\end{equation}
where the inequality stems from the fact that $\nerows_{\per}\le\nactions$, which implies $(\nerows_{\per}\nactions-\per)/(\nactions^{2}-\per) \le \nerows_{\per}\nactions/\nactions^{2}$.

The next two results provide a control on the difference between the processes $\necolumns_{\per}$ and $\tnecolumns_{\per}$.
\begin{lemma}
\label{le:lim-P-Dt}    
Fix $\paralpha\in(0,1)$ and choose $\per=\lfloor\nactions^{1+\paralpha}\rfloor$.
Then
\begin{equation}
\label{eq:lim-P-Dt}
\lim_{\nactions\to\infty}\Prob\parens*{\nsteps_{\per} < 6\nactions^{2\paralpha}} = 1.
\end{equation}
\end{lemma}

\begin{proof}
Consider a sequence of independent Bernoulli random variables $\parens*{\bern_{j}}$ such that $\Prob\parens*{\bern_{j}=1}=j/\nactions^{2}$.
Then \cref{eq:P-nsteps} implies
\begin{equation}
\label{eq:P-D-Bern}
\Prob\parens*{\nsteps_{\per} \ge 6\nactions^{2\paralpha}}
\le\Prob\parens*{\sum_{j=1}^{\per}\bern_{j} \ge 6\nactions^{2\paralpha}}.
\end{equation}
Since
\begin{equation}
\label{eq:Sum-Bern}
\Expect\bracks*{\sum_{j=1}^{\per}\bern_{j}}
\le
\frac{\per^{2}}{\nactions^{2}}
=
\nactions^{2\paralpha},
\end{equation}
the multiplicative Chernoff bound 
\citep[see, \eg][theorem~4.4]{MitUpf:CUP2017}
implies that
\begin{equation}
\label{eq:P-3alpha} 
\Prob\parens*{\nsteps_{\per} \ge 6\nactions^{2\paralpha}}
\le2^{-6\nactions^{2\paralpha}}.
\end{equation}
\end{proof}
\begin{corollary}
\label{co:Ct-tilde-Ct}    
For any $\paralpha\in(0,1/2)$ and $\per \le \nactions^{1+\paralpha}$, we have
\begin{equation}
\label{eq:Ct-tilde-Ct}    
\frac{\necolumns_{\per}-\tnecolumns_{\per}}{\nactions}
\xrightarrow[\nactions\to\infty]{\Prob} 0.
\end{equation}
\end{corollary}
Similarly to what was done with the process $\necolumns_{\per}$, we introduce a simplified process $\tgreene_{\per}$ which serves as an approximation of $\greene_{\per}$. More precisely, we set $\tgreene_0=0$ and define, for all $\per\in[\nactions^2]$
\begin{equation}
\label{eq:def-tgreene}
\tgreene_{\per}-\tgreene_{\per-1}=\ind_{\nerows_{\per}-\nerows_{\per-1}=1}\,\ind_{\tnecolumns_{\per}-\tnecolumns_{\per-1}=1}.
\end{equation}
Notice that the evolution of $\greene_{\per}$ can be obtained by replacing $\tnecolumns_{\per}$ and $\tnecolumns_{\per-1}$ with  $\necolumns_{\per}$ and $\necolumns_{\per-1}$, respectively, in \cref{eq:def-tgreene}.

\subsection{An \ac{ODE} approximation}
\label{suse:ODE-approx}

We now want to show that the triple $\parens*{\nerows_{\per},\tnecolumns_{\per},\tgreene_{\per}}$, when properly re-scaled (in space-time) by $\nactions$ can be approximated by the solution of an \ac{ODE} system.
To do this, we need the following lemma:

\begin{lemma}
\label{le:three-functions}  
Let $\funca,\funcb,\funcc \colon [0,1]^{3}\to[0,1]$ be defined as follows:
\begin{equation}
\label{eq:three-functions-def}
\funca(\snerows,\snecolumns,\sgreene) 
= 
1-\snerows,\quad
\funcb(\snerows,\snecolumns,\sgreene) 
= 
1-\snecolumns,\quad
\funcc(\snerows,\snecolumns,\sgreene) 
= 
(1-\snerows)(1-\snecolumns).
\end{equation}
Then, for every $\per\le\nactions^{3/2}$, we have
\begin{align}
\label{eq:rt-funca}
\abs*{\Expect\bracks*{\nerows_{\per+1}-\nerows_{\per} \mid \sigmaf_{\per}}
- \funca\parens*{\frac{\nerows_{\per}}{\nactions},\frac{\tnecolumns_{\per}}{\nactions},\frac{\tgreene_{\per}}{\nactions}}}
&\le
\frac{1}{\sqrt{\nactions}},\\
\label{eq:rt-funcb}
\abs*{\Expect\bracks*{\tnecolumns_{\per+1}-\tnecolumns_{\per} \mid \sigmaf_{\per}}
- \funcb\parens*{\frac{\nerows_{\per}}{\nactions},\frac{\tnecolumns_{\per}}{\nactions},\frac{\tgreene_{\per}}{\nactions}}}
&\leq
\frac{1}{\sqrt{\nactions}},\\
\label{eq:rt-funcc}
\abs*{\Expect\bracks*{\tgreene_{\per+1}-\tgreene_{\per} \mid \sigmaf_{\per}}
- \funcc\parens*{\frac{\nerows_{\per}}{\nactions},\frac{\tnecolumns_{\per}}{\nactions},\frac{\tgreene_{\per}}{\nactions}}}
&\le
\frac{2}{\sqrt{\nactions}}.
\end{align}
\end{lemma}

\begin{proof}
To prove \cref{eq:rt-funca} notice that, by  \cref{eq:success-row-prob}, 
\begin{equation}
\Expect\bracks*{\nerows_{\per+1}-\nerows_{\per} \mid \sigmaf_{\per}}=\frac{(\nactions-\nerows_{\per})\nactions}{\nactions^{2}-\per}.
\end{equation}
For $0\le \per\le \nactions^{3/2}$, we have
\begin{align*}
\abs*{\Expect\bracks*{\nerows_{\per+1}-\nerows_{\per} \mid \sigmaf_{\per}}
- \funca\parens*{\frac{\nerows_{\per}}{\nactions},\frac{\tnecolumns_{\per}}{\nactions},\frac{\tgreene_{\per}}{\nactions}}}=\abs*{\frac{(\nactions-\nerows_{\per})\nactions}{\nactions^{2}-\per}-\left(1-\frac{\nerows_\per}{\nactions}\right)}=\frac{\nactions-\nerows_{\per}}{\nactions}\cdot \frac{\per}{\nactions^2-\per}\leq \frac{\per}{\nactions^2}\leq\frac{1}{\sqrt{\nactions}}\,.
\end{align*}

To prove \cref{eq:rt-funcb} recall \cref{eq:success-col-prob} and note that
\begin{align*}
\Expect\bracks*{\tnecolumns_{\per+1}-\tnecolumns_{\per}\mid \sigmaf_{\per}}&=
\Expect\bracks*{\tnecolumns_{\per+1}-\tnecolumns_{\per}\mid \succrow_{\per+1}=1,\sigmaf_{\per}}\Prob(\succrow_{\per+1}=1\mid\sigmaf_{\per})
\\&\quad+\Expect\bracks*{\tnecolumns_{\per+1}-\tnecolumns_{\per}\mid  \succrow_{\per+1}=0,\noresample_{\per+1}=1,\sigmaf_{\per}}\Prob(\succrow_{\per+1}=0,\noresample_{\per+1}=1\mid\sigmaf_{\per})
\\&\quad+\Expect\bracks*{\tnecolumns_{\per+1}-\tnecolumns_{\per}\mid  \succrow_{\per+1}=0,\noresample_{\per+1}=0,\sigmaf_{\per}}\Prob(\succrow_{\per+1}=0,\noresample_{\per+1}=0\mid\sigmaf_{\per})
\\&=\succprobcol_{\per+1}\Prob(\succrow_{\per+1}=1\mid\sigmaf_{\per})+\succprobcol_{\per+1}\Prob(\succrow_{\per+1}=0,\noresample_{\per+1}=1\mid\sigmaf_{\per})+0
\\&=\succprobcol_{\per+1}(1-\Prob(\succrow_{\per+1}=0, \noresample_{\per+1}=0\mid\sigmaf_{\per}))\,.
\end{align*}
Hence, by \eqref{eq:P-nsteps}, for $0\leq \per\leq \nactions^\frac32$ we have
\begin{align*}
\abs*{\Expect\bracks*{\tnecolumns_{\per+1}-\tnecolumns_{\per} \mid \sigmaf_{\per}}
- \funcb\parens*{\frac{\nerows_{\per}}{\nactions},\frac{\tnecolumns_{\per}}{\nactions},\frac{\tgreene_{\per}}{\nactions}}}&=\abs*{\Expect\bracks*{\tnecolumns_{\per+1}-\tnecolumns_{\per}\mid \sigmaf_{\per}}-\succprobcol_{\per+1}}
\\&\leq\Prob(\succrow_{\per+1}=0,\noresample_{\per+1}=0\mid\sigmaf_{\per})\leq \frac{\per}{\nactions^2}\leq \frac{1}{\sqrt{\nactions}}\,.
\end{align*}

Finally, \cref{eq:rt-funcc} follows from \cref{eq:def-tgreene}. Indeed, being the sequences $(\nerows_\per)_{\per\geq 0}$ and $(\tnecolumns_\per)_{\per\geq 0}$ independent, we have
\begin{align*}
&\abs*{\Expect\bracks*{\tgreene_{\per+1}-\tgreene_{\per} \mid \sigmaf_{\per}}
- \funcc\parens*{\frac{\nerows_{\per}}{\nactions},\frac{\tnecolumns_{\per}}{\nactions},\frac{\tgreene_{\per}}{\nactions}}}=
\\&=
\abs*{\Expect\bracks*{\nerows_{\per+1}-\nerows_{\per} \mid \sigmaf_{\per}}\Expect\bracks*{\tnecolumns_{\per+1}-\tnecolumns_{\per} \mid \sigmaf_{\per}}
- \funca\parens*{\frac{\nerows_{\per}}{\nactions},\frac{\tnecolumns_{\per}}{\nactions},\frac{\tgreene_{\per}}{\nactions}}\funcb\parens*{\frac{\nerows_{\per}}{\nactions},\frac{\tnecolumns_{\per}}{\nactions},\frac{\tgreene_{\per}}{\nactions}}}
\\&\leq \abs*{\Expect\bracks*{\nerows_{\per+1}-\nerows_{\per} \mid \sigmaf_{\per}}
- \funca\parens*{\frac{\nerows_{\per}}{\nactions},\frac{\tnecolumns_{\per}}{\nactions},\frac{\tgreene_{\per}}{\nactions}}}+\abs*{\Expect\bracks*{\tnecolumns_{\per+1}-\tnecolumns_{\per} \mid \sigmaf_{\per}}
- \funcb\parens*{\frac{\nerows_{\per}}{\nactions},\frac{\tnecolumns_{\per}}{\nactions},\frac{\tgreene_{\per}}{\nactions}}}\leq \frac{2}{\sqrt{\nactions}}\,.
\end{align*}
\end{proof}

At this point, the main approximation tools which we will exploit in the following is summarized in the following result.
\begin{proposition}
\label{le:Warnke}  
Let $\parens*{\snerows_{\peralt},\snecolumns_{\peralt},\sgreene_{\peralt}}_{\peralt\ge 0}$ be the solution of the Cauchy system
\begin{equation}
\label{eq:Cauchy-Warnke}
\begin{cases}
\dot\snerows_{\peralt} = 1 - \snerows_{\peralt},\\
\dot\snecolumns_{\peralt} = 1 - \snecolumns_{\peralt},\\
\dot\sgreene_{\peralt} = (1 - \snerows_{\peralt})(1 - \snecolumns_{\peralt}),\\
\snerows_{0} = \snecolumns_{0} = \sgreene_{0} = 0,
\end{cases}
\end{equation}
that is,
\begin{equation}
\label{eq:sol-Cauchy-Warnke}
\snerows_{\peralt} = \snecolumns_{\peralt} = 1 - \expo^{-\peralt},\quad 
\sgreene_{\peralt} = \frac{1}{2} \parens*{1-\expo^{-2\peralt}}.
\end{equation}
Then, for any constant $\constT>0$, we have
\begin{equation}
\label{eq:Warnke}  
\Prob\parens*{\forall \per\le\constT\nactions,
\abs*{\frac{\nerows_{\per}}{\nactions}-\snerows_{\per/\nactions}}
+
\abs*{\frac{\tnecolumns_{\per}}{\nactions}-\snecolumns_{\per/\nactions}}
+
\abs*{\frac{\tgreene_{\per}}{\nactions}-\sgreene_{\per/\nactions}}
\le 
18 \expo^{2\constT}\nactions^{-1/3}}
\ge 
1 - 6\exp\braces*{-\frac{\nactions^{1/3}}{8\constT}}.
\end{equation}
\end{proposition}

\begin{proof}
The result can be deduced using the so-called \emph{Wormald differential equation method}  \citep[see][]{Wor:AAP1995}. 
For simplicity, we adopt here the same formulation as in  
\citet[theorem~2]{War:arXiv2019}, and use the bounds proved therein. 
With this method, the evolution of the space-time rescaled random process can be uniformly first-order approximated by the solution of an associated Cauchy problem. 

The main requirement is that the expected increment of the original process be controlled by certain function (which plays the role of the derivative in the associated Cauchy problem) up to a vanishingly small error, as we showed in
\cref{le:three-functions}. 
\end{proof}

We can now prove \cref{pr:asymptotic-epsilo-IC}.

\begin{proof}[Proof of \cref{pr:asymptotic-epsilo-IC}]
Fix $\varepsilon\in(0,1/2)$. If we define
\begin{equation}
\label{eq:epsilon-time-tilde}
\widetilde{\stime}_{\varepsilon\nactions} \coloneqq \inf \braces*{\per \ge 0 \colon \tgreene_{\per} = \floor{\varepsilon\nactions}},
\end{equation}
then, by \cref{le:Warnke}, we obtain
\begin{equation}
\label{eq:stopping-time-concentrates-tilde}
\frac{\widetilde{\stime}_{\varepsilon\nactions}}{\nactions} 
\xrightarrow[\nactions\to\infty]{\Prob}
\frac{1}{2}\log\parens*{\frac{1}{1-2\varepsilon}}.
\end{equation}
In light of \eqref{eq:stopping-time-concentrates-tilde}, to prove
\ref{it:pr:asymptotic-epsilo-IC-a} it suffices to show that, for all $\per \le \nactions^{7/6}$
\begin{equation}\label{eq:last}
\frac{\abs*{\tgreene_\per-\greene_{\per}}}{\nactions}\xrightarrow[\nactions\to\infty]{\Prob}0.
\end{equation}
Notice that \cref{eq:last} follows by \cref{co:Ct-tilde-Ct} and \cref{eq:def-tgreene}.
At this point, the convergence in \ref{it:pr:asymptotic-epsilo-IC-b} is a consequence of \ref{it:pr:asymptotic-epsilo-IC-a}, \cref{le:Warnke}, and \cref{co:Ct-tilde-Ct}.
\end{proof}


\section{Proofs of the main results}
\label{se:main-proofs}

To help the reader understand the idea behind the proof of \cref{th-BoA-espilon}, we consider an example that illustrates the behavior of a \ac{BRD} in its first steps.

\begin{example}
\label{ex:incremental-construction}
We consider the behavior of the incremental construction in a small game.
\cref{fig:pic1} shows it pictorially.
The probability that $\BRD(\argdot)$ travels on the exact path shown in \cref{fig:pic1} is
\begin{equation}
\label{eq:BRD-fig-1}
\Prob\parens*{\BRD(\per)=\profile_{\per}, \text{ for }\per=1,\dots,6 \mid \BRD(0)=\profile_{0}} =
\frac{1}{\nactions}\cdot
\frac{1}{2\nactions-1}\cdot
\frac{1}{3\nactions-2}\cdot
\frac{1}{4\nactions-4}\cdot
\frac{1}{5\nactions-6}\cdot
\frac{1}{6\nactions-9}
\end{equation}
To wit, $\BRD(1)=\profile_{1}$ if $\potential(\profile_{1})$ is the smallest in its column, which happens with probability $1/\nactions$.
Given this event, we have that $\BRD(2)=\profile_{2}$ if $\potential(\profile_{2})$ is the smallest of its row and the previously visited column, which happens with probability $1/(2\nactions-1)$.
Given $\BRD(1)=\profile_{1}$ and $\BRD(2)=\profile_{2}$, we have that $\BRD(3)=\profile_{3}$ if $\potential(\profile_{3})$ is the smallest of all previously visited rows column, which happens with probability $1/(3\nactions-2)$.
Given $\BRD(i)=\profile_{i}$, for all $i\in\braces*{1,2,3}$, we have that $\BRD(4)=\profile_{4}$ if $\potential(\profile_{4})$ is the smallest of all previously visited rows column, which happens with probability $1/(4\nactions-4)$.
And so on.
\begin{figure}[h!]
\centering
\includegraphics[width=10cm]{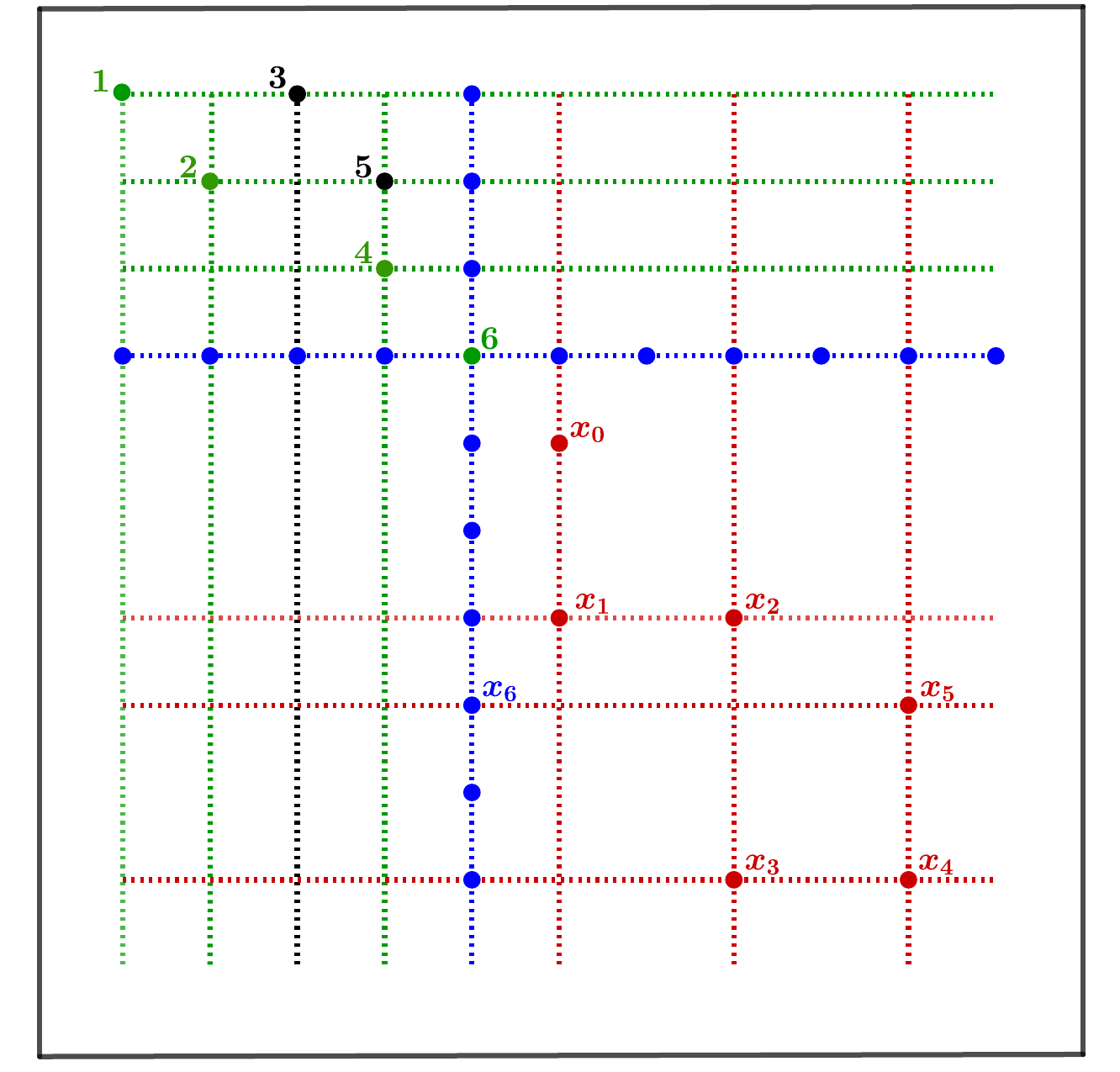}
\caption{In this example $\varepsilon\nactions=4$, $\stime_{\varepsilon\nactions}=6$, $\nerows_{\stime_{\varepsilon\nactions}}=4$ and $\necolumns_{\stime_{\varepsilon\nactions}}=5$.  
Green dots represent \acl{PNE}. 
When the \ac{BRD} reaches an action profile on a green dotted line, it reaches the \acl{PNE} on the same line. 
If the \ac{BRD} reaches an action profile on a blue dotted line, then it reaches the $\varepsilon\nactions$-th equilibrium. 
We start the process $\BRD(\argdot)$ at $\profile_{0}=(\nerows_{\stime_{\varepsilon\nactions}}+1,\necolumns_{\stime_{\varepsilon\nactions}}+1)=(5,6)$; the profiles $(\profile_{0},\profile_{1},\dots,\profile_{6})$ represent the trajectory of the \ac{BRD}.
The best-response dynamics reaches  the row of the $\varepsilon\nactions$-th equilibrium in $6$ steps.
Hence it reaches the $\varepsilon\nactions$-th equilibrium in $7$ steps. 
Strategy profiles on red dotted lines are explored by the best-response dynamics starting at $\profile_{0}$. }
	\label{fig:pic1}
\end{figure}
\end{example}

Consider a finite sequence of action profiles $\parens*{\profile_{\per}}_{\per=0}^{\Per}$ such that,  $\profile_{\per}=(\necolumns_{\per},\nerows_{\per})$, with 
$\necolumns_{\per}=\necolumns_{\per-1}$ for $\per$ odd and 
$\nerows_{\per}=\nerows_{\per-1}$ for $\per$ even.
Then
\begin{equation}
\label{eq:prob-BRD-T}
\Prob\parens*{\BRD(\per)=\profile_{\per}, \text{ for } \per=1,\dots,\Per} = 
\prod_{\per=1}^{\Per}
\frac{1}{\per\nactions-\funch(\per)},
\end{equation}
with
\begin{equation}
\label{eq:h(t)}  
\funch(\per)
\coloneqq 
\sum_{i=1}^{\per}\floor*{\frac{i}{2}}
\leq \frac{\per^2}{2}.
\end{equation}
Fix now $\varepsilon\in(0,1/2)$ and consider the $\varepsilon$-stopped incremental construction of \cref{de:epsilon-IC}.
First of all, notice that, by construction of the \ac{BRD}, all the elements in one column belong to the same basin of attraction. 
In particular, all profiles in  columns $1,\dots,\necolumns_{\stime_{\varepsilon\nactions}}-1$ will be attracted by equilibria $\equil_{\equilindex}$ with $\equilindex<\varepsilon\nactions$.
Profiles in column $\necolumns_{\stime_{\varepsilon\nactions}}$ will be attracted by $\equil_{\varepsilon\nactions}$.

We now focus on the columns $\necolumns_{\stime_{\varepsilon\nactions}}+1,\dots,\necolumns_{\nactions}$.
Consider the process $\BRD(\argdot)$ that starts at the profile $\profile_{0} = (\nerows_{\stime_{\varepsilon\nactions}}+1,\necolumns_{\stime_{\varepsilon\nactions}}+1)$.
The probability that, starting from $\profile_{0}$, the \ac{BRD} ends up in the $\varepsilon\nactions$-th \ac{NE} is obtained by summing the probability that the \ac{BRD} arrives in the $\varepsilon\nactions$-th \ac{NE} in a given number of steps. Notice that, starting at $\profile_0$, at least two steps are necessary to reach any \ac{NE}.  Therefore
\begin{equation}
\begin{split}
&\Prob\parens*{\BRD(2\nactions)  = \equil_{\varepsilon\nactions}
\mid     
\BRD(0)=\profile_{0},
\sigmaf_{\stime_{\varepsilon\nactions}}
} \\
&\qquad= \sum_{\nstepsl=0}^{2\nactions-2} \Prob\parens*{\BRD(\nstepsl+2)  = \equil_{\varepsilon\nactions}, \BRD(\nstepsl+1) \neq \equil_{\varepsilon\nactions}
\mid   
\BRD(0)=\profile_{0},
\sigmaf_{\stime_{\varepsilon\nactions}}
}.
\end{split}
\end{equation}
where the probabilities on the right-hand side, by the argument in \cref{eq:BRD-fig-1}, can be computed to be
\begin{equation}
\label{eq:P-BRD_tau-ell}
\begin{split}
&\Prob\parens*{\BRD(\nstepsl+2)  = \equil_{\varepsilon\nactions}
\mid   
\BRD(0)=\profile_{0},
\sigmaf_{\stime_{\varepsilon\nactions}}
} \\
&\quad=
\frac{1}{(\nstepsl+1)\nactions-\funch(\nstepsl+1)}
\prod_{\per=1}^{\nstepsl} 
\frac{\nactions-\nerows_{\stime_{\varepsilon\nactions}}\ind_{\braces*{\per \text{ odd}}}
-\necolumns_{\stime_{\varepsilon\nactions}}\ind_{\braces*{\per \text{ even}}}
-\floor*{\frac{\per}{2}}}{\per\nactions-\funch(\per)}
\, ,
\end{split}
\end{equation}
where $\funch(\argdot)$ is defined as in \ref{eq:h(t)}.
This implies
\begin{equation}
\label{eq:P-BRD_tau}
\begin{split}
&\Prob\parens*{\BRD(2\nactions)  = \equil_{\varepsilon\nactions}
\mid     
\BRD(0)=\profile_{0},
\sigmaf_{\stime_{\varepsilon\nactions}}
} \\
&\quad=
\sum_{\nstepsl=0}^{2\nactions-2}
\frac{1}{(\nstepsl+1)\nactions-\funch(\nstepsl+1)}
\prod_{\per=1}^{\nstepsl} \frac{\nactions-\nerows_{\stime_{\varepsilon\nactions}}\ind_{\braces*{\per \text{ odd}}}
-\necolumns_{\stime_{\varepsilon\nactions}}\ind_{\braces*{\per \text{ even}}}
-\floor*{\frac{\per}{2}}}
{\per\nactions-\funch(\per)}
\,,
\end{split} 
\end{equation}
where the product over an empty set of numbers is taken to be $1$.

Once we consider all possible positions for $\profile_{0}$, we obtain
\begin{equation}
\label{eq:expect-basin-tau}
\begin{split}
&\frac{1}{\nactions}\Expect\bracks*{\abs*{\BoA(\equil_{\varepsilon\nactions})} \mid \sigmaf_{\stime_{\varepsilon\nactions}}} \\
&\quad
=
1 + 
\parens*{\nactions - \necolumns_{\stime_{\varepsilon\nactions}}}\sum_{\nstepsl=0}^{2\nactions-2}
\frac{1}{(\nstepsl+1)\nactions-\funch(\nstepsl+1)}
\prod_{\per=1}^{\nstepsl} \frac{\nactions-\nerows_{\stime_{\varepsilon\nactions}}\ind_{\braces*{\per \text{ odd}}}
-\necolumns_{\stime_{\varepsilon\nactions}}\ind_{\braces*{\per \text{ even}}}
-\floor*{\frac{\per}{2}}}
{\per\nactions-\funch(\per)}
\,,
\end{split}
\end{equation}
where $1$ stands for the column in which the equilibrium $\equil_{\varepsilon \nactions}$ lies, while the other term takes into account each of the $\nactions-\necolumns_{\stime_{\varepsilon\nactions}}$ columns which are empty at the end of the construction, multiplied by the probability in \cref{eq:P-BRD_tau}. Notice, indeed, that by definition the first $\necolumns_{\stime_{\varepsilon\nactions}}-1$ columns are not in the basin of attraction of $\equil_{\varepsilon\nactions}$.
\begin{lemma}
\label{le:approx-basin}
For all $\varepsilon\in(0,1/2)$ and for all 
\begin{equation}
\label{eq:cK-rK-in}  
\snecolumns(\nactions),\snerows(\nactions) \in \naturals \cap \parens*{1-\sqrt{1-2\varepsilon}-\nactions^{-1/4},1-\sqrt{1-2\varepsilon}+\nactions^{-1/4}}\,,
\end{equation}
we have
\begin{equation}
\label{eq:approx-basin}   
\lim_{\nactions\to\infty}
1 +\sum_{\nstepsl=0}^{2\nactions -2}
\frac{\nactions - \snecolumns}{(\nstepsl+1)\nactions-\funch(\nstepsl+1)}
\prod_{\per=1}^{\nstepsl} \frac{\nactions-\snerows\ind_{\braces*{\per \textup{ odd}}}
-\snecolumns\ind_{\braces*{\per \textup{ even}}}
-\floor*{\frac{\per}{2}}}
{\per\nactions-\funch(\per)}
=\exp\braces*{\sqrt{1-2\varepsilon}},
\end{equation}
where we omitted the explicit dependence of $\snecolumns$ and $\snerows$ from $\nactions$.
\end{lemma}

\begin{proof}
Define, for $\zetasym\in\naturals$,
\begin{equation}
\label{eq:LKzeta}   
\LKzeta \coloneqq
\sum_{\nstepsl=0}^{2\nactions - 2}
\frac{\zetasym\nactions}{(\nstepsl+1)\nactions-\funch(\nstepsl+1)}
\prod_{\per=1}^{\nstepsl} \frac{\zetasym\nactions
-\floor*{\frac{\per}{2}}}
{\per\nactions-\funch(\per)}
\,.
\end{equation}
We have 
\begin{equation}
\label{eq:LKed}
\begin{split}
\sumprod_{\nactions}\parens*{\floor*{(1-\sqrt{1-2\varepsilon}-\nactions^{-1/4})}} 
&\le
\sum_{\nstepsl=0}^{2\nactions-2}
\frac{\nactions - \snecolumns}{(\nstepsl+1)\nactions-\funch(\nstepsl+1)}
\prod_{\per=1}^{\nstepsl} \frac{\nactions-\snerows\ind_{\braces*{\per \textup{ odd}}}
-\snecolumns\ind_{\braces*{\per \textup{ even}}}
-\floor*{\frac{\per}{2}}}
{\per\nactions-\funch(\per)}\\
&\le
\sumprod_{\nactions}\parens*{\ceil*{(1-\sqrt{1-2\varepsilon}+\nactions^{-1/4})}}
\,.
\end{split}
\end{equation}

Fixing a sequence $\Per_{\nactions}$ such that $\lim_{\nactions\to\infty}\Per_{\nactions}=\infty$
and $\Per_{\nactions}=\smalloh(\nactions)$, we get
\begin{equation}
\label{eq:LKzeta=} 
\begin{split}
\LKzeta 
&= 
\sum_{\nstepsl=0}^{2\nactions - 2}\frac{\zetasym}{(\nstepsl+1)-\bigoh(\nstepsl^{2}/\nactions)} \prod_{\per=1}^{\nstepsl} \frac{\zetasym - \bigoh(\per^{2}/\nactions)
}
{\per\nactions-\bigoh(\per/\nactions)}
\\
&=
\sum_{\nstepsl=0}^{\Per_{\nactions}}
\frac{\zetasym}{(\nstepsl+1)-\bigoh(\nstepsl^{2}/\nactions)}
\prod_{\per=1}^{\nstepsl} \frac{\zetasym - \bigoh(\per^{2}/\nactions)
}
{\per\nactions-\bigoh(\per/\nactions)}
 \\
&\quad+
\sum_{\nstepsl=\Per_{\nactions}+1}^{2\nactions - 2}
\frac{\zetasym}{(\nstepsl+1)-\bigoh(\nstepsl^{2}/\nactions)}
\prod_{\per=1}^{\nstepsl} \frac{\zetasym - \bigoh(\per^{2}/\nactions)
}
{\per\nactions-\bigoh(\per/\nactions)}
 \\
&\sim 
\sum_{\nstepsl=0}^{\Per_{\nactions}} \frac{\zetasym^{\nstepsl+1}}{(\nstepsl+1)!}
+\smalloh(1)\\
&\sim
\expo^{\zetasym}-1+\smalloh(1),
\end{split}
\end{equation}
which concludes the proof by choosing $\zetasym=\sqrt{1-2\varepsilon}\pm\nactions^{-1/4}$ and taking the limit $\nactions\rightarrow \infty$.
\end{proof}

\begin{proof}[Proof of \cref{th-BoA-espilon}]
Fix $\varepsilon\in(0,1/2)$ and $\delta\in(0,1)$. 
Consider the event
\begin{equation}
\label{eq:A-eps-delta}
\evented
\coloneqq \braces*{\abs*{\frac{\nerows_{\stime_{\varepsilon\nactions}}}{\nactions} -(1-\sqrt{1-2\varepsilon}) }\le \delta}\cap\braces*{\abs*{\frac{\necolumns_{\stime_{\varepsilon\nactions}}}{\nactions} -(1-\sqrt{1-2\varepsilon}) }\le \delta} ,
\end{equation}
and notice that $\evented$ is measurable with respect to $\sigmaf_{\stime_{\varepsilon \nactions}}$. 
Moreover, by \cref{pr:asymptotic-epsilo-IC} \cref{it:pr:asymptotic-epsilo-IC-b}
\begin{equation}
\label{eq:prob-A}
\lim_{\nactions\to\infty}\Prob(\evented)=1.
 \end{equation}
By the tower property
\begin{align}
    \Expect&\bracks*{\frac{1}{\nactions}\abs*{\BoA_{\nactions}(\equil_{\floor{\varepsilon\nactions}})}}=\Expect\bracks*{\Expect\bracks*{\frac{1}{\nactions}\abs*{\BoA_{\nactions}(\equil_{\floor{\varepsilon\nactions}})}\mid \sigmaf_{\stime_{\varepsilon \nactions}}}}\\
\label{eq:intermediate}
&\qquad=\Expect\bracks*{\Expect\bracks*{\frac{1}{\nactions}\abs*{\BoA_{\nactions}(\equil_{\floor{\varepsilon\nactions}})}\mid \sigmaf_{\stime_{\varepsilon \nactions}}}\ind_{\evented}}+\Expect\bracks*{\Expect\bracks*{\frac{1}{\nactions}\abs*{\BoA_{\nactions}(\equil_{\floor{\varepsilon\nactions}})}\mid \sigmaf_{\stime_{\varepsilon \nactions}}}\ind_{\evented^{\rm c}}}\\
    &\qquad\le\Expect\bracks*{\Expect\bracks*{\frac{1}{\nactions}\abs*{\BoA_{\nactions}(\equil_{\floor{\varepsilon\nactions}})}\mid \sigmaf_{\stime_{\varepsilon \nactions}}}\ind_{\evented}}+\nactions\Prob\parens*{\evented^{\rm c}},
\end{align}
where the inequality is due to the fact that $\BoA_{\nactions}(\equil_{\floor{\varepsilon\nactions}})\le\nactions^2$ almost surely. Notice further that, by \cref{le:Warnke}, it is possible to derive a quantitative version of \cref{eq:prob-A}, yielding
\begin{equation}
\lim_{\nactions\to\infty}\nactions\Prob\parens*{\evented^{\rm c}}=0.
\end{equation}
On the other hand, by neglecting the last term in \cref{eq:intermediate}, we obtain
\begin{align}
    \Expect&\bracks*{\frac{1}{\nactions}\abs*{\BoA_{\nactions}(\equil_{\floor{\varepsilon\nactions}})}}
    \ge\Expect\bracks*{\Expect\bracks*{\frac{1}{\nactions}\abs*{\BoA_{\nactions}(\equil_{\floor{\varepsilon\nactions}})}\mid \sigmaf_{\stime_{\varepsilon \nactions}}}\ind_{\evented}}.
\end{align}
The desired conclusion follows from these estimates and \cref{eq:prob-A}, since \cref{eq:expect-basin-tau} and \cref{le:approx-basin} imply
\begin{equation}
\label{eq:last-est}
   \lim_{\delta\to 0}\lim_{\nactions\to\infty} \Expect\bracks*{\Expect\bracks*{\frac{1}{\nactions}\abs*{\BoA_{\nactions}(\equil_{\floor{\varepsilon\nactions}})}\mid \sigmaf_{\stime_{\varepsilon \nactions}}}\ind_{\evented}}=\exp\braces*{\sqrt{1-2\varepsilon} }.
\end{equation}
\end{proof}

\begin{proof}[Proof of \cref{co:BoA-epsilon}]
The probability on the left-hand side of \cref{eq:BoA-epsilon-DF} can be rewritten as the expectation of
\begin{equation}
\label{eq:BRD<eK}  
\Prob\parens*{\BRD(2\nactions) = \equil_{j}, j \le \varepsilon\nactions \mid \sigmaf_{\nactions^{2}}}
=
\sum_{j=1}^{\varepsilon\nactions}\frac{\abs*{\BoA(\equil_{j})}}{\nactions^{2}}
\,.
\end{equation}
Therefore the convergence in distribution in \cref{eq:BoA-epsilon-DF} follows from the tower property and  \cref{th-BoA-espilon}. 
    To obtain the convergence of the expectation in  \cref{eq:BoA-epsilon-Expect} it is enough to use \cref{eq:BoA-epsilon-DF} and the fact that, by definition, $\rank(\BRD(2\nactions))/\nactions \le 1$.
\end{proof}

\begin{proof}[Proof of \cref{th:BoA-rank-potential}]
The process $\parens*{\nerows_{\per}}_{\per\ge 0}$ can be seen as a coupon-collector processes, where $\nerows_\per$ stands for the number of coupons collected by time $\per$. 
If we call $\coupont_{\genact}$ the time that is needed to collect the $\genact$-th coupon, then we have
\begin{equation}
\label{eq:expectation-coupon}  
\Expect\bracks*{\coupont_{\genact}}
=
\nactions(\harmonic_{\genact}-\harmonic_{\nactions-\genact}),
\end{equation}
where $\harmonic_{\genact}$ is the $\genact$-th harmonic number
\begin{equation}
\label{eq:harmonic}
\harmonic_{\genact}\coloneqq 
\sum_{i=1}^{\genact} \frac{1}{i}.   
\end{equation}
Moreover, for each $\genact\in\actions$, we have 
\begin{equation}
\label{eq:variance-coupon}   
\Var\bracks*{\coupont_{\genact}} \le 2\nactions^{2}.
\end{equation}
If, for some $\varepsilon\in(0,1)$, $\genact= \nactions-\nactions^\varepsilon$, then
\begin{equation}
\Expect[\coupont_{\genact}]\sim(1-\varepsilon)\nactions\log(\nactions),
\end{equation}
and Chebyshev inequality implies that,  for all $\delta>0$
\begin{equation}
\label{eq:coupon-Chebyshev}
\lim_{\nactions\to \infty} \Prob\left(\left|\frac{\coupont_{\genact}}{(1-\varepsilon)\nactions\log(\nactions)}-1\right|>\delta\right)=0\,,
\end{equation}
Notice that, the probability of having a Nash equilibrium with ranking larger than $\per$ is bounded above by the probability that the incremental construction in \cref{al:incremental-construction} does not stop before time $\per$, i.e.,
\begin{align}
\label{eq:P-lower-bound} 
\Prob\parens*{\max\braces*{\necolumns_{\floor*{(1+\delta)\nactions\log\nactions}},\nerows_{\floor*{(1+\delta)\nactions\log\nactions}}} = \nactions}
&\ge
\Prob\parens*{\nerows_{\floor*{(1+\delta)\nactions\log\nactions}} = \nactions}\\
&=\Prob\parens*{\coupont_{\nactions}\le\floor*{(1+\delta)\nactions\log\nactions} }\to 1,
\end{align}
where the asymptotic follows by \eqref{eq:coupon-Chebyshev}. This concludes the proof.
\end{proof}


\section{List of symbols}
\label{se:symbols}

\begin{longtable}{p{.13\textwidth} p{.82\textwidth}}

$\actA$ & action of player $\pA$\\

$(\eq\actA,\eq\actB)$ & \acl{NE}\\ 
$\evented$ & defined in \cref{eq:A-eps-delta}\\
$\actB$ & action of player $\pB$\\
$\BoA(\eq\actA,\eq\actB)$ & basin of attraction of $(\eq\actA,\eq\actB)$, defined in \cref{eq:BoA}\\
$\BRD$ & \acl{BRD}\\

$\snecolumns_{\per}$ & $1 - \expo^{-\peralt}$, defined in \cref{eq:sol-Cauchy-Warnke}\\

$\necolumns_{\per}$ & number of non-empty columns after placing the first $\per$ entries of $\potential$\\

$\succcol_{\per}$ & Bernoulli random variable such that 
$\Prob(\succcol_{\per}=1)=
\succprobcol_{\per}$, defined in \cref{eq:success-col-prob}\\

$\nsteps_{\per}$ & number of steps where \ref{it:srow0-noresample0} occurs\\
${\sigmaf_{\per}}$ & $ \sigma(\potential^{\per})$\\

$\sgreene_{\per}$ & $\frac{1}{2} \parens*{1-\expo^{-2\peralt}}$, defined in \cref{eq:sol-Cauchy-Warnke}\\

$\greene_{\per}$ & number of green entries after adding the first $\per$ entries of $\potential$\\
  
$\funch(\per)$ & $\sum_{i=1}^{\per}\floor*{\frac{i}{2}}$, defined in \eqref{eq:h(t)}\\

$\harmonic_{\nactions}$ & $\nactions$-th harmonic number\\

$\nactions$ & number of actions\\

$\actions$ & action set of each player\\
   
$\LKzeta$ & defined in \cref{eq:LKzeta}\\

$\ematrix_{\per}$ & sub-matrix of $\potential$ composed of rows $\braces*{1,\dots,\nerows_{\per}}$ and columns $\braces*{1,\dots,\necolumns_{\per}}$\\

$\NE_{\nactions}$ & set of \aclp{PNE}\\

$\snerows_{\per}$ & $1 - \expo^{-\peralt}$, defined in \cref{eq:sol-Cauchy-Warnke}\\
$\nerows_{\per}$ & number of non-empty rows after placing the first $\per$ entries of $\potential$\\

$\coupont_{\genact}$ & time needed to collect the $\genact$-th coupon\\
$\succrow_{\per}$ & Bernoulli random variable such that $\Prob(\succrow_{\per}=1)=
\succprobrow_{\per}$, defined in \cref{eq:success-row-prob}\\ 

$\noresampleprob_{\per+1} $ & $ 
1- \frac{\per}{\necolumns_{\per}\nactions}$, defined in \cref{eq:no-resamples} \\

$\noresample_{\per+1}$ & Bernoulli random variable such that $\Prob(\noresample_{\per+1}=1)=
\noresampleprob_{\per+1}$, defined in \cref{eq:no-resamples}\\ 


$\sampleX_{\per+1}, \sampleY_{\per+1}$ & random variables drawn uniformly without replacement on the empty cells of $\ematrix_{\per}$ \\

$\sampleZ_{\per+1}$ & random variable with uniform distribution on  $\bracks*{\necolumns_{\per}}$\\

$\funca$ & defined in \cref{eq:three-functions-def}\\
$\funcb$ & defined in \cref{eq:three-functions-def}\\
$\funcc$ & defined in \cref{eq:three-functions-def}\\


$\equil$ & \acl{NE}\\

$\outcome(\actA,\actB)$ & outcome of the profile $(\actA,\actB)$\\
$\succprobcol_{\per}$ & $\frac{(\nactions-\necolumns_{\per-1})}{\nactions}$, defined in \cref{eq:success-col-prob}\\
$\rank(\equil)$ & ranking of equilibrium $\equil$\\
$\succprobrow_{\per}$ & $\frac{(\nactions-\nerows_{\per})\nactions}{\nactions^{2}-\per}$, defined in \cref{eq:success-row-prob}\\
$\stime_{\varepsilon\nactions} $ &$ \inf \braces*{\per \ge 0 \colon \greene_{\per} = \floor{\varepsilon\nactions}}$, defined in \cref{eq:epsilon-time}\\
$\funcphi(\varepsilon)$ & $\exp\braces*{\sqrt{1-2\varepsilon}}$, defined in \cref{eq:BoA-epsilon}\\
$\funcPhi(\varepsilon)$ & $(1-\sqrt{1-2\varepsilon})\exp\braces*{\sqrt{1-2\varepsilon}}$, defined in \cref{eq:BoA-epsilon-DF}\\
$\potential$ & potential function, defined in \cref{eq:ordinal-potential}\\
$\eq\potential$ & $\braces*{\potential(\equil) \colon \equil \in \NE_{\nactions}}$, defined in \cref{eq:set-equilibrium-potential}\\
$\potential^{\per}$ & set of all entries added at times $\peralt \le \per$\\
$\prec_{\play}$ & preference relation of player $\play$\\

\end{longtable}

\subsection*{Acknowledgments}
Hlafo A. Mimun, Matteo Quattropani and Marco Scarsini are members of GNAMPA-INdAM.
Their work was partially supported by the  MIUR PRIN 2022EKNE5K ``Learning in markets and society''  and the INdAM GNAMPA project CUP\_E53C22001930001 ``Limiting behavior of stochastic dynamics in the Schelling segregation model.''


\bibliographystyle{apalike}
\bibliography{bibpotentialBRD.bib}

\end{document}